\documentclass[12pt,reqno]{amsart}
\usepackage[comma]{natbib}
\usepackage[a4paper,left=1in,right=1in,top=1.25in,bottom=1.25in]{geometry}
\usepackage{color}
\usepackage{graphicx}
\usepackage{amssymb}
\usepackage{enumerate}

\newtheorem{theorem}{Theorem}%[section]

\theoremstyle{definition}
\theoremstyle{remark}

\def\numberlikeadb{\global\def\theequation{\thesection.\arabic{equation}}}
\numberlikeadb

\begin{document}

\title[]{A graphical exploration of the relationship between parasite aggregation indices} 

\author[]{R. McVinish}
\address{School of Mathematics and Physics, University of Queensland}
\email{r.mcvinish@uq.edu.au}
\author[]{R.J.G. Lester}
\address{School of Biological Sciences, University of Queensland}


\begin{abstract}
    The level of aggregation in parasite populations is frequently incorporated into ecological studies.  It is measured in various ways including variance-to-mean ratio, mean crowding, the $k$ parameter of the negative binomial distribution and indices based on the Lorenz curve such as the Gini index (Poulin's D) and the Hoover index. Assuming the frequency distributions follow a negative binomial, we use contour plots to clarify the relationships between aggregation indices, mean abundance and prevalence. The contour plots highlight the nonlinear nature of the relationships between these measures and suggest that correlations are not a suitable summary of these relationships.
\end{abstract}

\maketitle

\section{Introduction}

Investigations into parasite population dynamics frequently require an indicator of the level of aggregation in the parasite population \citep{Tinsley:2020,KTCPGA:2022}. As the concept of aggregation in parasites is poorly defined \citep{Pielou:77, ML:2020}, aggregation has been measured in various ways. Commonly used indices include prevalence, the Variance-to-Mean Ratio (VMR), and the $k$ parameter of the negative binomial distribution. Closely related to VMR and $k$ are mean crowding and patchiness \citep{Lloyd:67} which can be seen as more direct measures of the competitive experience of parasites within a host \citep{WFL:2018}. Two other indices are derived from the Lorenz curve \citep{Lorenz:1905}, the most widely accepted quantification of inequality. \citet{Poulin:93} proposed using the Gini index \citep{Gini:1914}, which has since become widely used in parasitology \citep{Rod:2021, BB:2023, Matos:2023}.  The Hoover index (aka Pietra index) has more recently been proposed to measure parasite aggregation \citep{ML:2020, LB:2021}.  

This paper clarifies and extends our previous work on aggregation. It was stimulated by a recent paper by \citet{MPF:2023} which correlated aggregation indices with mean abundance and prevalence using simulated data.  We present a more accurate representation using ‘contour plots’, calculated directly from the parameters of the negative binomial distributions.  The plots provide a simple and more insightful way to comprehend the relationships.

\section{Contour plots}

The contours show combinations of two indices, specified on the vertical and horizontal axes, that give rise to similar values of the third index. Contour plots, developed in the 16th century \citep{MM:2017}, are widely used in other disciplines but rarely in parasitology (e.g. \citet{KTCPGA:2022}).

Our analysis assumed that parasite burden is adequately modelled by a negative binomial distribution \citep{Crofton:71, SGD:98, Poulin:2011, MPF:2023}. Following the typical practice in parasitology, we parameterised the negative binomial distribution in terms of mean abundance, $m$, and the parameter $k$ which controls the shape of the distribution. We did not make any assumption on the distribution of $m$ and $k$. We used the range of values for $m$ and $k$ suggested by the extensive data of \citet{SD:95}. Their values for $m$, $k$ and prevalence are superimposed on several of the contour plots as dot points.  

To construct a contour plot of an aggregation index against $m$ and $k$, we expressed the aggregation index as a function of $m$ and $k$. The population values of several indices can be expressed simply in terms of $m$ and $k$: \[\text{prevalence} = 1 - (k/(k+m))^{k},\] $\text{VMR} = 1 + m/k$ , $\text{mean crowding} = m +m/k$, and $ \text{patchiness} = 1 + 1/k$. The Gini and Hoover indices lack simple expressions in terms of $m$ and $k$ however, they can still be evaluated numerically. The Hoover index can be expressed in terms of $m$ and $k$ by applying \citet[Lemma 5.3.3]{Arnold:87},
\[
H=F(m;k,m)-F(m-1;k+1,m+m/k),
\]
where $F(x; m, k)$ is the cumulative distribution function of the negative binomial distribution with $k$ and mean $m$ evaluated at $x$.  Further details are given in the Appendix. The cumulative distribution function of the negative binomial distribution, $F$, is available in statistical packages such as R \citep{RCT:2023}. The Gini index can be expressed as
\[
G= \left(1+\frac{m}{k}\right)\, _{2}F_1 \left(k+1,\frac{1}{2},2;-4 \frac{m}{k} \left(1+\frac{m}{k}\right) \right),
\]
where $_{2}F_{1}$ is the Gaussian hypergeometric function \citep[equation 2.12]{Ramasubban:58}. This can be evaluated in R using the hypergeo package \citep{Hankin:2015}. Calculating indices directly from the parameters of the negative binomial distribution rather than using simulated data obviates the need to consider the uncertainty of estimates and the effects of different sample sizes. 

We also employed contour plots to examine the relationship between aggregation indices, $k$, and prevalence. This required first solving the equation \[\text{prevalence} = 1 – (k/(k+m))^{k}\] in terms of $k$ for each pair of $m$ and prevalence in the contour plot. This equation has a unique solution if $ m + \ln(1 - \text{prevalence}) > 0$. On the other hand, if $m + \ln(1 - \text{prevalence}) < 0$, there is no solution to the equation. The solution was found numerically using the uniroot function in R. The expressions for the aggregation indices in terms of $m$ and $k$ is then used to construct the contour plot. Regions of $m$ and prevalence that are inconsistent with a negative binomial distribution are represented as white in the contour plot.  

All contour plots were produced in R using the ggplot2 package \citep{Wickham:2016}. The values for $m$ and $k$ reported in \citet{SD:95} were heavily skewed and spanned several orders of magnitude with $m$ ranging between 0.1 and 5200 and $k$ ranging between 0.001 and 16.5. To make the plots clearer, log scaling has been applied to these variables. 

\section{Relationship between mean abundance, $k$, and prevalence}

The relationship between $m$, $k$, and prevalence in wild parasite populations has been examined by several authors with conflicting results \citep{Pennycuick:71, Scott:87, Poulin:93, SD:95, KTCPGA:2022}. While the expression of prevalence in terms of $m$ and $k$ is sufficiently simple to analyse, it is still instructive to construct the contour plot (Fig. \ref{fig:1} left).  In it, each colour represents a region of values of $m$ and $k$ that give rise to similar values of prevalence. 

We see that prevalence is increasing in both $m$ and $k$ leading to contours that are roughly L-shaped on the range of $m$ and $k$ plotted so prevalence is small when either $m$ or $k$ are small, and prevalence is large when both $m$ and $k$ are large. The contours also show that there is a non-linear relationship between $m$ and $k$ when prevalence is considered fixed. The contours become almost parallel to the horizontal axis as $k$ increases, a consequence of $\lim_{k\to\infty} \text{prevalence} = 1 - e^{-m}$. On the other hand, the contours continue to move left as $m$ increases, a consequence of $\lim_{m\to\infty} \text{prevalence} = 1$.  The contour plot shows that the rate at which prevalence approaches one as $m$ increases is slow when $k$ is small.

If we restrict our attention to a single-coloured band, i.e. those values of $m$ and $k$ giving rise to similar values of prevalence, we see that, after controlling for prevalence, there is a negative relationship between $m$ and $k$.  This relationship is forced by the negative binomial distribution, so it will hold true in natural systems to the extent that those systems are well modelled by the negative binomial distribution. The different widths of the contour lines show the non-linearity of the relationship between $m$, $k$ and prevalence.

The dot points represent estimates of $m$ and $k$ from the 269 parasite-host systems reported in \citet{SD:95}. Although several parasite-host systems lie in a region of very high prevalence (both $m$ and $k$ large) or very small prevalence (either $m$ or $k$ small), many others occupy a region of the parameter space where a moderate change in the parameter values would result in a significant change in prevalence assuming a negative binomial distribution.

As Shaw \& Dobson reported prevalences in their review, it is possible to compare these with the prevalence values implied by the negative binomial distribution (Fig. \ref{fig:1} right). In general, there is good agreement; most points within a given contour having the same colour.  This demonstrates the accuracy of the contour plots to interpret relationships in real life situations. The few points where the observed prevalences don’t agree with that determined by the negative binomial could be because these distributions did not conform to a negative binomial.  

\begin{figure}
    \centering
    \includegraphics[width=\textwidth]{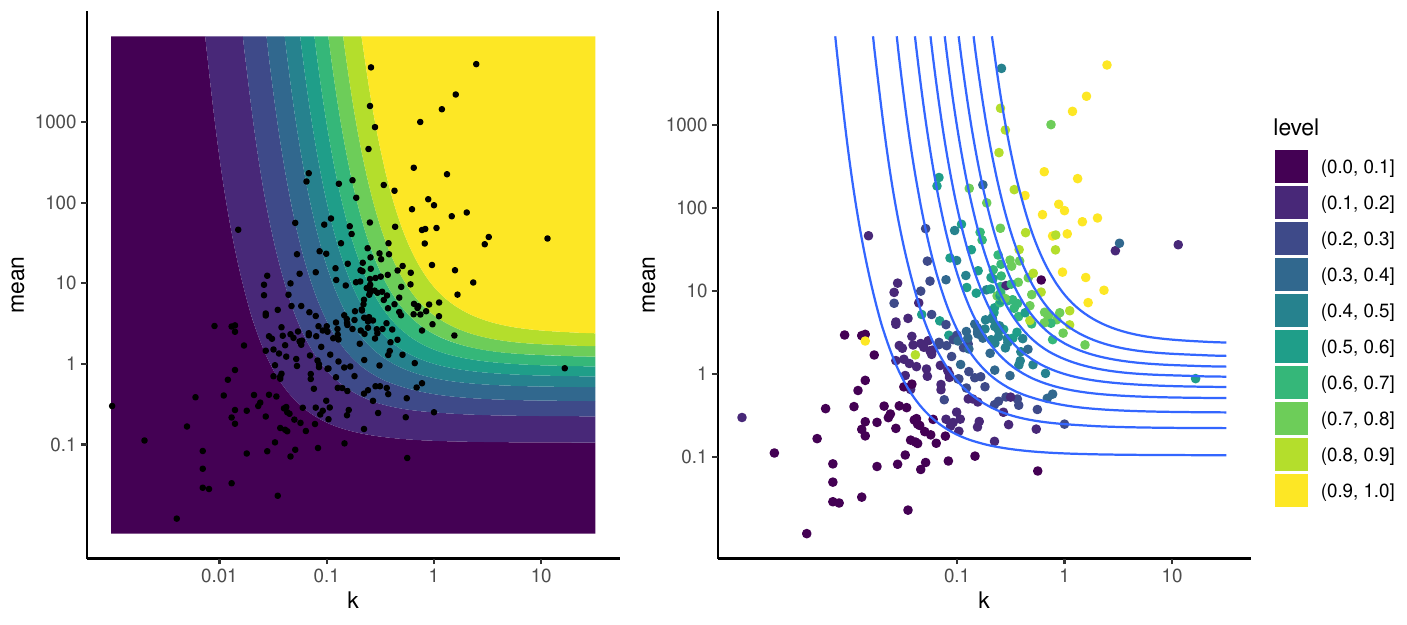}
    \caption{ (Left) Contour plot showing prevalence levels from zero to 1 for values of $m$ and $k$. Each colour band represents a region of values of $m$ and $k$ that give rise to similar values of prevalence. The axes for both the $m$ and $k$ are on the log scale. Dot points are actual data from \citet{SD:95}. (Right) Scatter plot of mean against $k$ with prevalence data from Shaw \& Dobson.  Lighter colours indicate higher prevalence level.  The lines are from the contour plot (left).  In general, there is good agreement between the observed prevalences and the contours.  Exceptions may be from samples that did not conform to a negative binomial.}
    \label{fig:1}
\end{figure}

\section{Relationship of Hoover \& Gini indices with mean abundance, $k$, and prevalence}

Contour plots of the Hoover index and Gini index as functions of $m$ and $k$ are shown in Fig \ref{fig:2} left and right. The contour plots are qualitatively very similar and share some similarities with the contour plot of prevalence (Fig. \ref{fig:1}). Both Hoover and Gini indices decrease in both $m$ and $k$, taking values close to one when either $m$ or $k$ were small, and taking values close to zero when both $m$ and $k$ were large. The contours are L-shaped becoming almost parallel to the horizontal axis as $m$ increases and almost parallel to the vertical axis as $k$ increases.

\begin{figure}
    \centering
    \includegraphics[width=\textwidth]{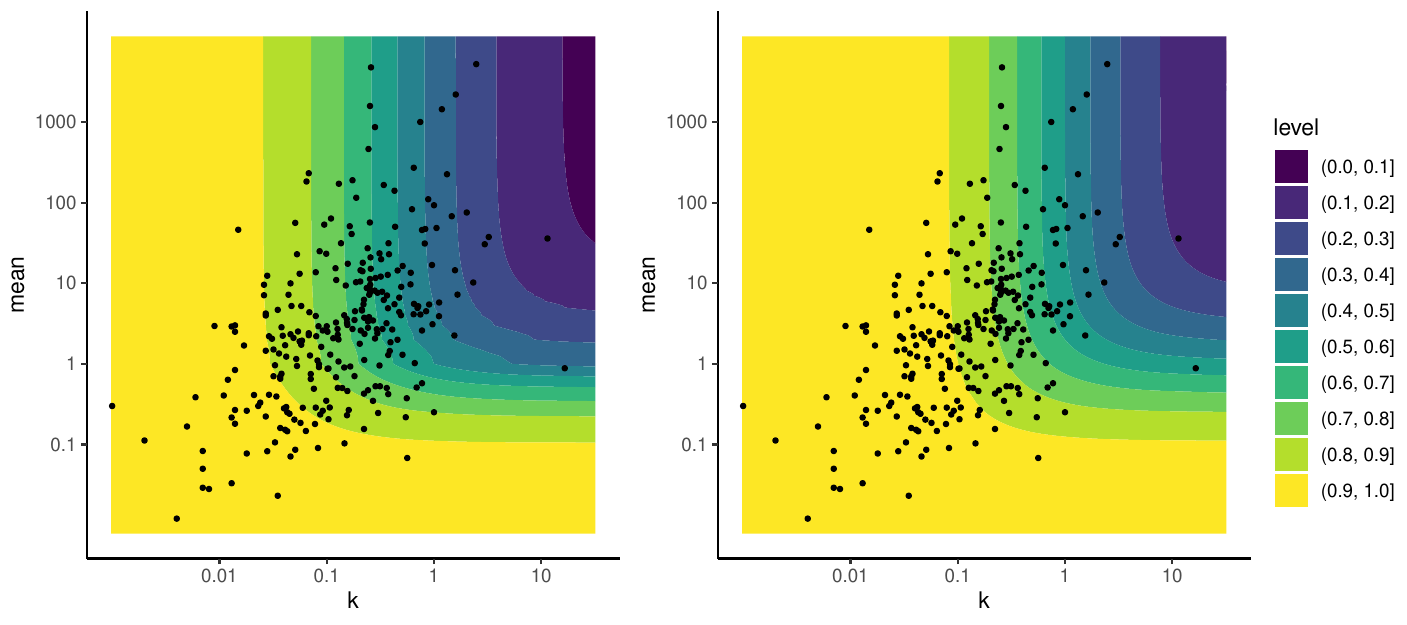}
    \caption{Contour plots of Hoover index (left) and Gini index (right) as functions of $m$ and $k$. Contour lines are shown from 0 to 1, i.e. from least aggregated (darkest band) to most aggregated (yellow band). Comparing the graphs, the Gini index is always larger than the Hoover index and has a smaller range over the region of values for $m$ and $k$ observed in wild populations (dot points)}
    \label{fig:2}
\end{figure}

The plots show both indices display some stability over a wide range of $m$ and $k$. Restricting our attention to the Hoover index (Figure \ref{fig:2} left), we see that for $m>5$ the value of the index is largely determined by the size of $k$. For $ m< 5$ the value is less affected by $k$ but more affected by $m$, as indicated by the number of contours crossed as $m$ decreases. For example, starting from $k=1$ and $m=6$, as $m$ decreases the value of the index increases quickly crossing several contours from 0.4 to 1. On the other hand, when $m$ increases from the same point (1,6) the index stays in the same colour band and there is little change in the Hoover value (0.4 to 0.5). For many of the parasite-host systems reported in \citet{SD:95}, shown on the figure as dot points, an increase in $m$, that is moving the points vertically on the contour plot, does not appear to impact the Hoover index since the point would remain in the same-coloured region. On the other hand, in many of the samples, a moderate change in $k$, that is moving the point horizontally, has a large impact on the Hoover index. Similar behaviour is observed in the contour plot of the Gini index (Fig. \ref{fig:2} right), with the Gini index appearing to be even less affected by changes in $m$.

There are two noticeable differences between the contour plots for the Hoover and Gini indices (Fig. \ref{fig:2}). Firstly, the Gini index is always larger than the Hoover index \citep{Taguchi:68} \citep[Section 5.7]{Arnold:87}. This causes the Gini index to have a smaller range over the region of values for $m$ and $k$ observed in wild populations. Specifically, for the values of $m$ and $k$ reported in \citet{SD:95}, the Gini index exceeds 0.9 in 42\% (113/269) of cases compared to 20\% (54/269) of cases exceeding 0.9 Hoover index. Second, the contours of the Hoover index are not smooth, unlike those of the Gini index. The bumps that occur on the contours of the Hoover index occur at integer values of the mean, the most prominent occurring when the mean is 1. These bumps quickly become much less noticeable as the mean increases.

The contour plots of the Gini and Hoover indices exhibit greater differences when considered as functions of $m$ and prevalence (Fig. \ref{fig:3}). First, unlike the Gini index, the contour lines of the Hoover index are parallel to the vertical axis when $m$ is less than one. As noted by \citet{MPF:2023}, when all infected hosts harbour infrapopulations larger than or equal to the overall mean, the Hoover index is equal to one minus prevalence. For the negative binomial distribution, this implies the Hoover index is equal to one minus prevalence when the mean is less than or equal to one. Second, there is less variability in the widths of the contours for the Hoover index compared to the Gini index. This suggests the dependence of the Hoover index on prevalence is more regular. At a given $m$, a change of 0.1 in the prevalence will have roughly the same effect on the value of the Hoover index, regardless of the initial value of prevalence. In contrast, much of the contour plot of the Gini index is coloured yellow, corresponding to values greater than 0.9. Values of the Gini index less than 0.6 are restricted to small region of the plot, indicating that small changes in prevalence in that region will result in a large change in the Gini index.

\begin{figure}
    \centering
    \includegraphics[width=\textwidth]{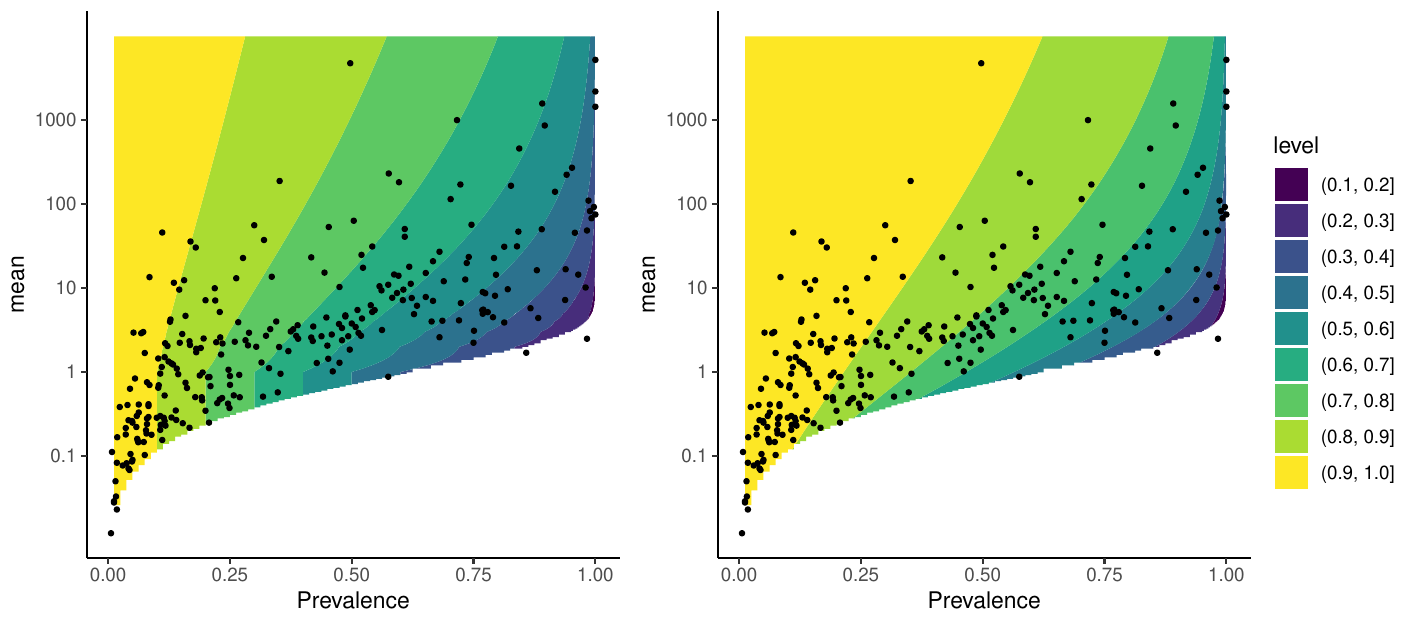}
    \caption{Contour plot of Hoover index (left) and Gini index (right) as functions of $m$ and prevalence. The contour lines of the Hoover index are parallel to the vertical axis when $m$ is less than one. The Gini index is less constrained.  However, with means above one, the Hoover index is more evenly spread than the Gini index. }
    \label{fig:3}
\end{figure}

The parasite data from Shaw and Dobson were taken from five taxonomic groups.  The data, divided into taxa, were superimposed on the plots of $m$ vs $k$ with contour lines of prevalence and Hoover index.  They did not show any obvious grouping.

\section{Lorenz order and the negative binomial distribution}

Both the Hoover and Gini indices are seen in Figure \ref{fig:2} to be decreasing functions of $m$ and $k$, as is 1 - prevalence (Figure \ref{fig:1}). This behaviour is due to how these indices relate to the Lorenz curve and how the parameters $m$ and $k$ affect the Lorenz curve of the negative binomial distribution.

The Lorenz curve of a distribution with cumulative distribution function $F$ is given by 
\[
L(u) =  \frac{\int^{u}_{0} F^{-1} (y)\, dy}{m}, \quad u \in [0,1],
\]
where $m$ is the mean of the distribution and $F^{-1}(x) = \sup_{x} \{x:F(x)\leq y\}$ for $y\in (0,1)$ (Gastwirth, 1971). In our context, the Lorenz curve describes the proportion $u$ of the host population that is infected with a proportion $L(u)$ of the parasite population. When all hosts have the same parasite burden, the Lorenz curve is given by $L(u) = u$ for all $u$ in [0,1]. This is called the egalitarian line. Several indices can be defined in terms of the Lorenz curve. Specifically, the Gini index is twice the area between the Lorenz curve and the egalitarian line, and the Hoover index is the greatest vertical distance between the Lorenz curve and the egalitarian line. Even $1 -\text{prevalence}$ can be viewed as the largest value of $u$ such that $L(u) = 0$.
The Lorenz curve induces a partial ordering of distributions. Assume $F_{A}$ and $F_{B}$ are two distribution functions with finite means. If the Lorenz curve of $F_{A}$ is greater than the Lorenz curve of $F_{B}$ for all $u$, then we say that $F_{A}$ is smaller than $F_{B}$ in the Lorenz order and write $F_{A} \leq F_{B}$. This ordering corresponds to the notion of aggregation put forward by \citet{Poulin:93,ML:2020}. From their connections with the Lorenz curve, we see that if $F_{A} \leq F_{B}$, then the Gini and Hoover indices as well as $1-\text{prevalence}$ will be smaller for $F_{A}$ than for $F_{B}$.

The following result shows that the negative binomial distribution decreases in the Lorenz order as $m$ increases and as $k$ increases. 

\begin{theorem} \label{Thm:lorenz1}
 Let $\mathsf{NB}(m,k)$ denote the negative binomial distribution with parameters $m$ and $k$. If $m_1 < m_2$, then 
\[
\mathsf{NB}(k,m_{2}) \leq_{L} \mathsf{NB}(k, m_{1}).
\] 
If $k_1 < k_2$, then 
\[
\mathsf{NB}(k_{2},m) \leq_{L} \mathsf{NB}(k_{1},m).
\]
\end{theorem}
The proof is provided in the Appendix.

The above result explains why Gini and Hoover indices and 1 – prevalence are all decreasing functions of $m$ and $k$. Figure \ref{fig:2} also shows that that the contours of both the Gini and Hoover indices become parallel with the axes. This is due to the limiting behaviour of the negative binomial distribution. Depending on how the parameters are allowed to vary, it is known that the negative binomial distribution will converge to either a Poisson distribution or a gamma distribution \citep{AdlC:94}. Fixing $m$ and letting $k$ increase, the negative binomial distribution converges to a Poisson distribution with mean $m$. This causes the contour lines to become parallel with the horizontal axis as $k$ increases. Similarly, fixing $k$ and letting $m$ increase, an appropriately scaled negative binomial distribution converges to a Gamma distribution with shape and rate parameters both equal to $k$. Since the Gini and Hoover indices are scale invariant \citep[Section 3.1]{Arnold:87}, these indices approach their respective values for a Gamma ($k$, $k$) distribution as $m$ increases. This causes the contour lines to become parallel with the vertical axis as $m$ increases.

\section{Discussion}

In choosing the index to use to measure aggregation, those based on Lorenz curves seem to be the favoured, such as the Hoover and Gini.  The Gini returns closer values over a wider range of means, $k$ and prevalence compared to the Hoover, making differences less discernible.  The Hoover has a biological interpretation and may be easier to calculate.  When mean abundances are below one, the Hoover index has restricted values whereas the Gini has no such restriction, suggesting that Gini may be preferred in such a situation. Nevertheless, both indices provide a figure that seems to measure the same phenomenon, a phenomenon that is still undefined.

The contour graphs provide an easily interpreted demonstration of the effects of the various parameters on the Hoover and Gini indices.  These could be deduced by an analysis of the formulae used to calculate the indices but this is not straightforward; indices do not correlate with a particular parameter.  When applying an index to compare aggregation between samples or species, it is useful to know which parameter is having the greatest effect on the index.  The contour graphs provide the answer.

In producing the graphs we calculated indices directly from the parameters of the negative binomial distribution rather than using simulated data as done by \citet{MPF:2023}.  This obviated the need to consider the uncertainty of estimates and the effects of different sample sizes. Our results demonstrated the deterministic functional relationships between the aggregation indices, and the parameters, mean abundance and prevalence. The relationships were not linear indicating that correlation and principal components analysis may not be the best methods to analyse the relationships \citep{MPF:2023}.

Listing the advantages and disadvantages of Hoover and Gini indices, \citet[Table 2]{MPF:2023} describe them as having the disadvantages of being “strongly negatively correlated with prevalence” and “weakly negatively correlated with mean abundance.” In contrast, the $k$ parameter of the negative binomial distribution and patchiness are described as having the advantages of being “not necessarily correlated with mean abundance” and “only weakly correlated with prevalence.”  These comments ignore the fact that the negative binomial distribution, and hence any index computed on that distribution, is completely specified by the mean and prevalence. In other words, the dependence of any index on mean and prevalence is perfectly deterministic. In fact, the dependence on any pair of quantities that can be used to parameterise the negative binomial distribution, like $m$ and $k$ is perfectly deterministic.

\citet{MPF:2023} argue that the Gini index is to be preferred over the Hoover index on the basis that Hoover index equals one minus prevalence when the mean is less than or equal to one whereas the Gini index has no such restriction. To decide between the Hoover and Gini indices, if one must choose, then the relationship between these indices and $m$, $k$ and prevalence need to be considered more closely. Our contour plots (Fig. \ref{fig:2} \& \ref{fig:3}) have shown other differences in the behaviour of the Gini and Hoover indices. Compared to the Gini index, the Hoover index has a greater range over the region of values for $m$ and $k$ (or prevalence) observed in wild populations and has more regular dependence on prevalence. Given these properties and the Hoover index’s clear biological interpretation, we argue that the Hoover index should be preferred over the Gini index, at least when $m$ is greater than one.

Our analysis has used contour plots to examine how the Gini and Hoover indices are affected by changes in $m$, $k$, and prevalence. This approach could, in principle, be applied to construct contour plots from any three indices, provided two of these can be used to parameterise the negative binomial distribution. For example, one could construct a contour plot of the Gini index as a function of VMR and mean crowding as both $m$ and $k$ can be expressed in terms of VMR and mean crowding: 
\[
m = \text{mean crowding} - \text{VMR} + 1
\]
and
\[
k = \frac{\text{mean crowding}}{ \text{VMR} -1} -1.
\]
The contour plot could then be constructed using the expression for the Gini index in terms of $m$ and $k$ given in Section 2. Further application of contour plots may unravel other complex relationships in ecological parasitology.

\appendix

\section{Hoover index of the negative binomial distribution}

In parasitology the negative binomial distribution is usually parameterised in terms of the the mean $ m $ and $ k $. The probability mass function is then
\[
f(x;k,m) =  {k+x-1 \choose k-1}  \left(\frac{k}{k+m} \right)^{k} \left(\frac{m}{k+m} \right)^{x}, \quad x \in \mathbb{N}_{0},
\]
and we write $ \mathsf{NB}(k,m)$. Let $F(\cdot;k,m)$ denote the cumulative distribution function of the $ \mathsf{NB}(k,m)$ distribution. The first moment distribution of the $ \mathsf{NB}(k,m)$ distribution, $F^{(1)}(\cdot;k,m)$, is 
\[
F^{(1)}(x;k,m) = \frac{\sum_{y\leq x} y\, f(y;k,m)}{m}.
\]
For any non-negative integer $x$
\begin{align*}
    \frac{x\, f(x)}{m} & = \frac{x}{m} {k+x-1 \choose k-1}  \left(\frac{k}{k+m} \right)^{k} \left(\frac{m}{k+m} \right)^{x} \\
    & = \frac{x}{m}\frac{(k+x-1)!}{(k-1)! x!} \left(\frac{k}{k+m} \right)^{k} \left(\frac{m}{k+m} \right)^{x}  \\
    %& = x\frac{(k+x-1)!}{(k-1)! x!} \left(\frac{k}{k+m} \right)^{k} \frac{1}{m} \left(\frac{m}{k+m} \right)^{x}  \\
    %& = \frac{(k+x-1)!}{(k-1)! x!} \left(\frac{k}{k+m} \right)^{k} \frac{1}{k+m} \left(\frac{m}{k+m} \right)^{x-1}  \\
    & = \frac{(k+x-1)!}{k! (x-1)!} \left(\frac{k}{k+m} \right)^{k+1} \left(\frac{m}{k+m} \right)^{x-1} \\
    & = \frac{(k+x-1)!}{k! (x-1)!} \left(\frac{k(1+1/k)}{(k+m)(1+1/k)} \right)^{k+1} \left(\frac{m(1+1/k)}{(k+m)(1+1/k)} \right)^{x-1} \\
    %& = \frac{(k+1 + (x-1) -1)!}{k! (x-1)!} \left(\frac{k +1 }{k+1 + (m+m/k)} \right)^{k+1} \left(\frac{m+m/k}{k+1 +(m+m/k)} \right)^{x-1} \\
    & = {(k+1) + (x-1) +1 \choose (k+1) -1} \left(\frac{k +1 }{k+1 + (m+m/k)} \right)^{k+1} \left(\frac{m+m/k}{k+1 +(m+m/k)} \right)^{x-1},
\end{align*}
which is the probability mass function of the $ \mathsf{NB}(k+1,m + m/k)$ distribution evaluated at $x-1$. Hence,
\[
F^{(1)}(x;k,m)  = F(x-1;k+1,m+m/k).
\]
\citet[Lemma 5.3.3]{Arnold:87} states that the Hoover index can be expressed as
\[
H = F(m;k,m) - F^{(1)}(m;k,m).
\]
Hence,
\[
H=  F(m;k,m) - F(m-1;k+1, m+m/k).
\]

\section{Proof of Theorem \ref{Thm:lorenz1}}

We first recall the definition of convex order, which is closely related to the Lorenz order \citep[subsection 3.A.1]{SS:07}. \\

{\bf Definition:} For random variables $X$ and $Y$ such that $ \mathbb{E} \, \phi(X) \leq \mathbb{E} \, \phi(Y)$ for all convex functions $ \phi : \mathbb{R} \to \mathbb{R}$ for which the expectations exist. Then we say that $ X$ is smaller than $Y $ in the convex order, denoted $ X \leq_{\rm{cx}} Y$ 
\\

The convex order relates to the Lorenz order in the sense that 
\[
\frac{X}{\mathbb{E} X} \leq_{\rm{cx}} \frac{Y}{\mathbb{E} Y}
\]
if and only if $ X \leq_{L} Y$, provided the expectations exist \citep[equation 3.A.33]{SS:07}  or \citep[Corollary 3.2.1]{Arnold:87}.

\begin{proof}[Proof of Theorem \ref{Thm:lorenz1}]
For part (a), let $ X_{2} \sim \mathsf{NB}(k,m_{2}) $. Conditional on $X_{2}$, let $ X_{1} \sim \mathsf{Binomial}(X_{2},m_{1}/m_{2})$. Then $ X_{1} \sim \mathsf{NB}(k, m_{1}) $. As $\mathbb{E} (X_{1}|X_{2}) = (m_{1}/m_{2}) X_{2} $ and $ \mathbb{E}X_{1}  = (m_{1}/m_{2})\, \mathbb{E}X_{2}$, \citep[Theorem 3.4]{Arnold:87} implies $ (m_{1}/m_{2})\, X_{2} \leq_{L} X_{1}$. Since the Lorenz order is invariant under a change of scale, $ \mathsf{NB}(k,m_{2}) \leq_{L} \mathsf{NB}(k,m_{1}) $.

For part (b), standard conditioning arguments show that if $ (N,\ t\geq 0)$ is a standard Poisson process and $ \Lambda \sim \mathsf{Gamma}(\alpha,\beta)$ (Gamma distribution with shape parameter $\alpha$ and rate parameter $\beta$), then $ N_{\Lambda} \sim \mathsf{NB}(\alpha, \alpha/\beta) $. Let $ T_{i} \sim \mathsf{Gamma}(k_{i}, k_{i}/m) $. 

It is known that for every convex function $\phi$, $ \mathbb{E} \phi(N_{t}) $ is a convex it $t$ \cite[Proposition 2]{Schweder:1982}. If we can show that $ T_{1} \leq_{cx} T_{2}$, then the result will follow from \citet[Theorem 3.A.21]{SS:07}.

By construction $\mathbb{E} T_{1} = \mathbb{E} T_{2}$. Let $ g_{i}$ be the probability density function of $ \mathsf{Gamma}(k_{i}, \beta_{i}) $. Then $ T_{1} \leq_{cx} T_{2} $ if $g_{2} - g_{1} $ exhibits exactly two sign changes in the sequence  +,  -, + \cite[Theorem 3.A.44]{SS:07}. As the $\log $ function is increasing, $ \log g_{2}  - \log g_{1}$ has the same sequence of sign changes as $ g_{2} - g_{1}$. Then
\begin{align*}
    \lefteqn{\log g_{2}(x) - \log g_{1}(x)} & \\
    & = (k_{2}-1)\log x - \frac{k_{2}}{m} x - \left((k_{1} -1) \log x - \frac{k_{1}}{m} x \right) + C \\
    & = (k_{2} - k_{1}) \log x - \frac{(k_{1} - k_{2})}{m} x + C,
\end{align*}
where $ C$ is depends on $ k_{1}, k_{2} $ and $ m $ but not $ x$. There must be at least one sign change since both $g_{1}$ and $g_{2} $ integrate to one. For $ k_{2} > k_{1}$ this function is concave so there must be two sign changes. As this function is positive for $ x \to 0 $ and $ x \to \infty$ when $ k_{2} > k_{1}$ we have shown $ T_{1} \leq_{cx} T_{2}$.  This completes the proof.
\end{proof}

\end{document}